\documentclass[10pt,final,journal,twocolumn,twoside,letterpaper]{IEEEtran}
\usepackage{cite}
\usepackage{graphicx}
\usepackage{amsmath}
\usepackage{times}
\usepackage{latexsym}
\usepackage{graphicx}
\usepackage{bm}
\usepackage{amssymb}
\usepackage[center]{caption2}
\usepackage{stfloats}
\usepackage{cases}
\usepackage{array}
\usepackage{setspace}
\usepackage{fancyhdr}
\usepackage{algorithm}
\usepackage{algorithmic}

\usepackage{amsmath}
\usepackage{graphicx}
\usepackage{subfigure}

\usepackage{xcolor}

\usepackage{flushend}

\usepackage{fancyhdr}
\usepackage{layout}
\allowdisplaybreaks[4]

\newcaptionstyle{mystyle1}{%
  \centering TABLE  \captiontext \par}
\captionstyle{mystyle1}
\newcaptionstyle{mystyle2}{%
\captionlabel $.$ \captiontext \par}
\captionstyle{mystyle2}
\newcaptionstyle{mystyle3}{%
 \centering \captionlabel $.$ \: \captiontext \par}
\captionstyle{mystyle3}

\begin{document}

\title{Energy Efficient Non-Cooperative  Power Control  in Small Cell Networks}


\author{\normalsize Yanxiang~Jiang,~\IEEEmembership{Member,~IEEE}, Ningning~Lu, Yan~Chen,~\IEEEmembership{Senior~Member,~IEEE}, Fuchun~Zheng,~\IEEEmembership{Senior~Member,~IEEE}, Mehdi~Bennis,~\IEEEmembership{Senior~Member,~IEEE}, Xiqi~Gao,~\IEEEmembership{Fellow,~IEEE}, and~Xiaohu~You,~\IEEEmembership{Fellow,~IEEE}
\thanks{Copyright (c) 2015 IEEE. Personal use of this material is permitted. However, permission to use this material for any other purposes must be obtained from the IEEE by sending a request to pubs-permissions@ieee.org.}
\thanks{Manuscript received April 27, 2016; revised December 1, 2016; accepted January 6, 2017.
This work was supported in part by the National 863 Project (2015AA01A709), the National Basic Research Program of China (973 Program 2012CB316004), the Natural Science Foundation of China (61221002, 61521061),
and the UK Engineering and Physical Sciences Research Council (EPSRC) under Grant EP/K040685/2.
The review of this paper was coordinated by Dr. W. Hamouda.
}
\thanks{Y. Jiang, N. Lu, X. Gao, and X. You are with the National Mobile Communications Research Laboratory,
Southeast University, Nanjing 210096, China (e-mail: yxjiang@seu.edu.cn, lnn0169@163.com, \{xqgao, xhyu\}@seu.edu.cn).}
\thanks{Y. Chen is with the School of Electronic Engineering, University of
Electronic Science and Technology of China, Chengdu 611051, China (e-mail: eecyan@uestc.edu.cn).}
\thanks{F. Zheng is with the National Mobile Communications Research Laboratory, Southeast University, Nanjing 210096, China, and the Department of Electronics, University of York, York, YO10 5DD, UK (e-mail: fzheng@ieee.org).}
\thanks{M. Bennis is with the Department of Communications Engineering, University of
Oulu, Oulu 90014, Finland (e-mail: bennis@ee.oulu.fi).}
\thanks{Digital Object Identifier 10.1109/TVT.2017.2673245.}
}


\maketitle

\thispagestyle{fancy}

\begin{abstract}
 In this paper, energy efficient power control for small cells underlaying a macro cellular network is investigated.
 We formulate the power control problem in self-organizing small cell networks as a non-cooperative game, and propose a distributed energy efficient power control scheme,
 which allows the small base stations (SBSs) to take individual decisions for  attaining the Nash equilibrium (NE) with minimum information exchange.
 Specially, in the non-cooperative power control game, a non-convex optimization problem is formulated for each SBS to maximize their energy efficiency (EE). By exploiting the properties of parameter-free fractional programming and the concept of perspective function, the non-convex optimization problem for each SBS is transformed into an equivalent constrained convex optimization problem. Then, the constrained convex optimization problem is converted into an unconstrained convex optimization problem by exploiting the mixed penalty function method. The inequality constraints are eliminated by introducing the logarithmic barrier functions and the equality constraint is eliminated by introducing the quadratic penalty function.
 {We also theoretically show the existence and the uniqueness of the NE in the non-cooperative power control game.}
 Simulation results show remarkable improvements in terms of EE by using the proposed scheme.
\end{abstract}

\begin{keywords}
 Energy efficiency, power control, non-cooperative game theory, perspective function, mixed penalty function.
\end{keywords}

\section{Introduction}
  {
  With the rapid expansion of wireless {communications} networks, tremendous spectrum efficiency (SE) and energy efficiency (EE)
  improvement is required for  5G mobile communication systems.}
   {As an expected feature of 5G systems, small}  cell networks (SCNs) \cite{Hoydis}, composed of a large number of {densely deployed} small cells with small coverage area and low transmission power, are expected to offer higher {SE and EE. }
   SCNs require self-organization, self-learning and intelligent decision making at  small cell base stations (SBSs),
     which can be randomly deployed by the operators or by the users in the hot-spot areas of city or rural locations. 

 Due to the large number of SBSs deployed in SCNs, conventional centralized power control schemes which require cooperation among all base stations (BSs) may not be practical.
   Most of the existing works on distributed power control for SCNs aimed at improving the sum throughput or decreasing the interference \cite{Carfagna,Chandrasekhar,Hong, Huang,Ma,Nie}.
  In \cite{Carfagna}, the authors proposed a distributed scheme based on pricing mechanism to maximize the sum-rate.
  A distributed utility was proposed to alleviate the cross-tier interference at the macrocell from cochannel femtocells in \cite{Chandrasekhar}.
  In \cite{Hong}, a payoff function was formulated to improve the fairness.
  In \cite{Huang}, the authors formulated a net utility function considering both the  gain of throughput and the punishment of interference.
  In \cite{Ma}, the authors proposed a novel power control scheme to guarantee the target sigal to interference plus noise ratios (SINRs) of the macrocell users, and make as many femtocell users as possible to achieve their target SINRs.
  {In \cite{Nie}, cooperative game theory was exploited to deal with the co-tier interference of the small cells.}

   With the explosive growth in data traffic, the energy consumption of wireless infrastructures and devices has increased greatly.
The {growing} energy consumption further brings about large electricity and maintenance bills for network operators,
and tremendous carbon emissions into the environment.
Consequently, improving the EE has become an important and urgent task.
   A large amount of works has recently been devoted to {investigating} the maximization of EE in wireless communications systems \cite{Isheden,Ng,Helmy,Jiang}.
   In general, the EE maximization problems are non-convex optimization problems,
   which can be transformed into  equivalent convex optimization problems by using fractional programming \cite{Dinkelbach}.
   In \cite{Isheden}, 
   the authors unified various approaches for the EE optimization problem for a typical  scenario in wireless {communications} systems,
   and discussed three {types} of fractional programming solutions in detail, including parametric convex program, parameter-free convex program and dual program.
   In \cite{Ng}, 
   the resource allocation for energy efficient secure communication in an orthogonal frequency-division multiple-access (OFDMA) downlink network was investigated,
   and the considered non-convex optimization problem was transformed into a convex optimization problem by exploiting the properties of parametric fractional programming.
   In \cite{Helmy}, 
   the authors proposed an energy efficient power control scheme for a multi-carrier link over a frequency selective fading channel with a delay-outage probability constraint,
   and derived the global optimum solution using parameter-free fractional programming.
   In \cite{Jiang}, by using fractional programming and exterior penalty function, an efficient iterative joint resource allocation and power control scheme was proposed to maximize the
   EE of device-to-device (D2D) communications.
   However, all the above research works on EE maximization have mainly focused on non-SCNs {considering} only one BS, which may not be applicable in SCNs.
   { Furthermore, the sheer number of densely deployed SBSs makes distributed power control more desirable.}
{For the existing research works on power control for SCNs, there are only a few studies concerning the optimization of EE.}
{In \cite{Haddad}, the authors introduced a Stackelberg game-theoretic framework for heterogeneous networks which enables both the small cells and the macro cells to strategically decide on their downlink
power control policies. In \cite{Yang}, a bargaining cooperative game framework for interference-aware power coordination was proposed.}

Motivated by the aforementioned discussions,  we propose a non-cooperative game theory based power control scheme for SCNs,
where each SBS  maximizes its own EE by performing power control independently and distributively.
By exploiting the properties of parameter-free fractional programming and the concept of perspective function, the non-convex optimization problem for each SBS player is transformed into an equivalent constrained convex optimization problem.
Then, this constrained convex optimization problem is further converted into an unconstrained convex optimization problem by exploiting the mixed penalty function method.
The inequality constraints are eliminated by introducing the logarithmic barrier functions and the equality constraint is eliminated by introducing a quadratic penalty function.
The existence and uniqueness of the Nash equilibrium (NE) in the non-cooperative power control game are further proved theoretically.

\section{System Model}

   We consider a two-tier downlink cellular network consisting of one macrocell and \(K\) small cells as illustrated in Fig. \ref{fig1}.
\begin{figure}[!b]
\centering
\includegraphics[width=0.45\textwidth]{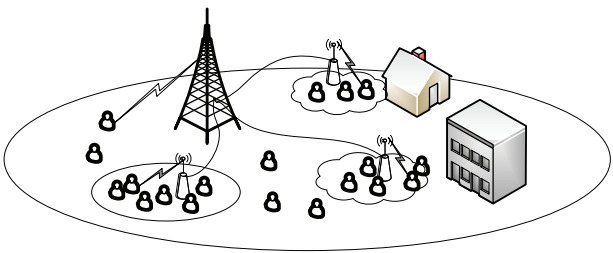}
\captionstyle{mystyle2}
\caption{System model of SCNs.}
\label{fig1}
\end{figure}
Let $\mathcal{K}=\{1,2,\cdots,K\}$ denote the set of the SBSs.
Let $\mathcal{N}_0={\rm{\{ 1,2,}} \cdots, {{N_0\} }}$ and $\mathcal{N}_k=\{1,2,\cdots,N_k\}$  denote the set of the macrocell user equipments (MUEs) and the set of  the small cell user equipemnts (SUEs) served by the $k$th SBS, respectively.
Let $\mathcal{N}=\{1,2,\cdots,N\}$ denote the set of the resource blocks (RBs).
Here, one RB refers to one time-frequency resource block which includes one time slot and 12 subcarriers\cite{Jiang}.
For efficient spectrum sharing, we assume that both the marcocell and small cells can utilize all the available RBs.
We also assume that the resource allocation of all the MUEs and SUEs has already been carried out.

Let \(P_k^i\) denote the transmit power of  the $k$th SBS to its corresponding SUE over the $i$th RB,
\(P_0^i\) the transmit power of the MBS to its corresponding MUE over the $i$th RB,
and  \(P_l^i\)  the transmit power of the $l$th SBS to its corresponding SUE over the $i$th RB.
Let \(\boldsymbol{p}_{  k}=[P_k^1,P_k^2, \cdots, P_k^i, \cdots, P_k^{N}]^T\).
Let {\(R_k(\boldsymbol{p}_{  k})\) }denote the transmission rate of the $k$th SBS.
Then, it can be expressed as,
\begin{multline}\label{Rate}
{R_k}(\boldsymbol{p}_{  k})  \\
= W \sum\limits_{i = 1}^{N}{\log_2 \left(1 + \frac{{H_{k,k}^iP_k^i}}{{H_{0,k}^iP_0^i{\rm{ + }} {\sum\limits_{l \ne k, l = 1}^{K}} {H_{l,k}^iP_l^i} + {N_0} }} \right)},
\end{multline}
%
where $W$ denotes the bandwidth of one RB,
\(H_{k,k}^i\)  the channel gain  from the $k$th SBS to its corresponding SUE over the $i$th RB,
\(H_{0,k}^i\)  the interference channel gain from the MBS to the corresponding SUE over the $i$th RB associated with the $k$th SBS,
\(H_{l,k}^i \)  the interference channel gain from the $l$th SBS to the corresponding SUE over the $i$th RB associated with the $k$th SBS,
and \({N_0}\)  the noise power.

Let \(\text{EE}_k\) denote the EE of the $k$th SBS. Then, it can be expressed as,
\begin{equation}\begin{split}\label{EE}
{\text{EE}_k}= \frac{{{R_k}}}{{P_k^c{\rm{ + }} \frac{1}{\sigma} \sum\limits_{i = 1}^{N} {P_k^i} }},
\end{split}\end{equation}
where $P_k^c$ denotes the circuit power consumption of the $k$th SBS,
and it is a power offset that is independent of the radiated power, but derived from signal processing, and the related circuit power, etc\cite{Jiang}.
$\sigma$ denotes the inefficiency of the power amplifier, and $0<\sigma\leq 1$.
Correspondingly, the EE of the entire system can be expressed as \cite{Jiang}
\begin{equation}\begin{split}\label{EEs}
{\text{EE}_s}= \frac{{\sum\limits_{k = 1}^{K}{R_k}}}{\sum\limits_{k = 1}^{K}\left({P_k^c{\rm{ + }}   \frac{1}{\sigma} \sum\limits_{i = 1}^{N} {P_k^i} }\right)}.
\end{split}\end{equation}

{It should be pointed out here that we use the heterogeneous feature of SCNs and consider multiple small cells and multiple users in the system model.
Both the cross-tier interference from the marcrocell to the considered small cell and the co-tier interference from the other small cells to the considered small cell are taken into account.
}


\section{Proposed Energy Efficient Non-cooperative Power Control Scheme}

In this section, we propose a distributed energy efficient power control scheme for SCNs based on non-cooperative game theory \cite{Nash}.
In the formulated non-cooperative power control game, each SBS is allowed to take individual decisions for attaining the NE with minimum information exchange.
The fractional objective function in (\ref{EEs}) is a {non-concave} function which is hard to solve directly.
{Just as} we pointed out in \cite{Jiang}, a brute force approach is generally required for obtaining a global optimal solution.
However, such a method has an exponential complexity with respect to the number of RBs and the number of cells.
Therefore, we propose an {efficient} method to solve this challenging problem.


\subsection{Problem Formulation}


Using microeconomic concepts, we assume that all the considered SBSs participate in a \(K\)-player non-cooperative power control game \(G = [\mathcal{K} ,\{ {{\boldsymbol{p}}_k}\} ,\{ {u_k}( \cdot )\} ]\),
where \(\{\boldsymbol{p}_{  k}\}\) denotes  the set of the transmit power vectors of the $k$th SBS,
and \({u_k}(\cdot)\) denotes the utility function of the $k$th SBS to be maximized.
Let \({\boldsymbol{p}}_{  -k}=[{\boldsymbol{p}}^T_{  1},{\boldsymbol{p}}^T_{  2},...,{\boldsymbol{p}}^T_{  k-1},{\boldsymbol{p}}^T_{  k+1},...,{\boldsymbol{p}}^T_{  K}]^T\) denote the transmit power vector of the $K-1$ other SBSs.
Then, for all the considered SBSs, the non-cooperative power control game can be formulated as follows,
 \begin{equation}\label{Game model}
\mathop {\max }\limits_{{\boldsymbol{p}}_k^{} \in C_k} {u_k}\left( {{\boldsymbol{p}}_k^{},{\boldsymbol{p}}_{ - k}^{}} \right),\text{ for each SBS in }\mathcal{K},
\end{equation}
where \(C_k\) is the feasible region of \({\boldsymbol{p}}_{ k}\). In our case,
{the utility function that we consider is EE, i.e.,}
${u_k}\left( {{\boldsymbol{p}}_k^{},{\boldsymbol{p}}_{ - k}^{}} \right) = \text{EE} _ {k}$.
Let \(P_t\) denote the maximum power consumption of each SBS, and \(R_t\) the minimum rate requirement of each SBS.
By considering the transmit power constraint and the transmission rate constraint, the above optimization problem for the $k$th SBS can be expressed in the following equivalent form,
\begin{equation}\begin{split}\label{Nonconvex problem}
            &\mathop {\max }\limits_{{\boldsymbol{p}}_k} {u_k}\left( {{\boldsymbol{p}}_k^{},{\boldsymbol{p}}_{ - k}^{}} \right)\\
\text{s.t.    }
            &{C_{k,1}}:{P^i _{k}}\ge 0, i = 1,2,\cdots,N\\
            &{C_{k,2}}:\sum\limits_{{i} = 1}^{N} {P_{k}^i}  \le {P_{t}}\\
            & {C_{k,3}}: {R_k(\boldsymbol{p}_k)}  \ge {R_{t}}.
\end{split}\end{equation}
{where \({C_{k,1}}\) and \({C_{k,2}}\) are the transmission power constraints for the  $k$th SBS, and \({C_{k,3}}\) guarantees the target rate requirement of the $k$th SBS.}
{In the following, where no confusion is caused, we also use $C_{k,1}, C_{k,2},$ and  $C_{k,3}$ to denote the sets of transmit powers which meet the corresponding constraints.
Therefore,}  \({C_{k,1}} \cap {C_{k,2}} \cap C{}_{k,3}=C_k\).
Note that, as in \cite{Jiang}, all the subcarriers within one RB are characterized by the same channel gains in our formulated
EE optimization problem, and hence, no single-user diversity is exploited.
However, the assumption of constant channel gain in the allocated subcarriers in one RB makes the power control
problem more practical, despite it representing a suboptimal solution.

In the above non-cooperative power control game,  each SBS optimizes its individual utility that depends on the transmit powers of the $K-1$ other SBSs.
It is necessary to calculate the equilibrium point wherein each SBS achieves its maximum utility
{ conditioned on the transmit powers of the  $K-1$ other SBSs.}
Such an operating point in the optimization problem (\ref{Game model}) is called an NE \cite{Nash}.
Let \( {{\boldsymbol{p}}^{\rm{*}}_k} \) denote the $k$th SBS's best response to the strategy \( {{\boldsymbol{p}}^{\rm{*}}_{-k}} \) 
specified for the $K-1$ other SBSs at the NE.
Then, \( {{\boldsymbol{p}}^{\rm{*}}_k} \) can be expressed as follows,
\begin{equation}\begin{split}\label{Nash point}
{\boldsymbol{p}}_k^* = \arg \mathop {\max }\limits_{{\boldsymbol{p}}_k^{} \in C_k} {u_k}\left( {{\boldsymbol{p}}_k^{},{\boldsymbol{p}}_{ - k}^*} \right).
\end{split}\end{equation}
It means that no SBS can obtain a unilateral profit by deviating from the NE, i.e.,
\begin{multline}\label{3-01}
{u_k}\left( {{\boldsymbol{p}}^{\rm{*}}_k,{\boldsymbol{p}}^{\rm{*}}_{ - k}} \right) \ge {u_k}\left( {{\boldsymbol{p}}_k^{},{\boldsymbol{p}}^{\rm{*}}_{ - k}} \right),\\
 \text{for every feasible strategy}\  {{\boldsymbol{p}}_k} \ \text{in} \  C_k.
\end{multline}
Accordingly, the goal of our considered optimization problem is to find the NE in the non-cooperative power control game \(G\).

\subsection{Iterative Algorithm for EE Maximization}

In the following, we {propose} an iterative power tuning strategy to reach the NE based on (\ref{Nash point}) in the considered non-cooperative power control game.
Let $n$ denote the number of iterations and $\epsilon$ denote the convergence threshold.
Let ${\boldsymbol{p}}_{k}(n)$ and ${\boldsymbol{p}}_{-k}(n)$ denote the transmission vectors of the $k$th SBS and the $K-1$ other SBSs for the $n$th iteration, respectively.
Then, the iterative algorithm based on the non-cooperative power control game can be summarized in Algorithm 1.
In each iteration, the transmission power strategy \({\boldsymbol{p}}_{k}(n+1)\) of  the $k$th SBS is a best response to the transmission power strategy \({\boldsymbol{p}}_{-k}(n)\) of the $K-1$ other SBSs. 
The iterative algorithm converges to  the only solution if and only if there exists one unique NE in the non-cooperative power control game,
and the NE can be reached when each SBS's transmission power strategy is sufficiently close to that in the previous iteration.

\begin{algorithm}[!h]
\caption{The iterative algorithm based on the non-cooperative power control game}
\begin{itemize}
\item Step 1: Initialization: {\(n = 1\),  $\epsilon $, \({\boldsymbol{p}}_k(1)\), and \({\boldsymbol{p}}_{-k}(1)\).}
\item Step 2:  For each SBS (i.e., $k=1,2, \cdots, K$), calculate
 ${{\boldsymbol{p}}_k}(n + 1) = \arg \mathop {\max }\limits_{{\boldsymbol{p}}_k^{} \in C_k} {u_k}\left[ {{\boldsymbol{p}}_k^{},{\boldsymbol{p}}_{ - k}}(n) \right], \   k \in \mathcal{K}$.
\item Step 3: {Set \(\beta=\sum\nolimits_{k = 1}^{K}|{u_k}\left[ {\boldsymbol{p}}_k(n+1), {\boldsymbol{p}}_{-k}(n+1) \right]-{u_k}\left[ {\boldsymbol{p}}_k(n),  {\boldsymbol{p}}_{-k}(n) \right] |\).
              If $\beta$ $ \geq $ $\epsilon$, set \(n=n+1\).}
\item Step 4: Repeat the steps $2\thicksim3$ until $\beta$ $<$ $\epsilon$.
\end{itemize}
\end{algorithm}

\subsection{Problem Equivalence}

The optimization problem in (\ref{Nonconvex problem}) involves a fractional objective function which  is non-convex,
and it requires great computational burden even for a small sized system.
Therefore, we resort to a more computationally efficient scheme below to solve this challenging problem.

\subsubsection{{Convex Transformation}}

Exploiting the properties of  parameter-free fractional programming \cite{Pfree} and the concept of perspective function \cite{Boyd},
we propose to transform the original non-convex optimization problem in (\ref{Nonconvex problem}) into an equivalent convex optimization problem.
Apply the variable transformation
\begin{align*}
y_k^0 &=  \frac{1}{{P_k^c{\rm{ + }}\sum\nolimits_{i = 1}^{N} {P_k^i} }}, \\
y_k^i &= \frac{P_k^i}{{P_k^c{\rm{ + }}\sum\nolimits_{i = 1}^{N} {P_k^i} }}. 
\end{align*}
Define
\begin{align*}
{{\boldsymbol{y}}_k^{'}} &= \left[ \frac{y_k^1}{ y_k^0}, \frac{y_k^2}{y_k^0}, \cdots, \frac{y_k^i}{y_k^0}, \cdots, \frac{y_k^{N}}{y_k^0} \right]^T, \\
{{\boldsymbol{y}}_k} &= \left[ {y_k^0,y_k^1, \cdots, y_k^i, \cdots, y_k^{N}} \right]^T.
\end{align*}
Let \(R_k^{'}({{\boldsymbol{y}}_k})=R_k({{\boldsymbol{y}}_k}^{'})\), and define $\zeta({\boldsymbol{y}}_k){\rm} = {y_k^0}R_k^{'}({\boldsymbol{y}}_k)$.
Then, the following {optimization problem} can be formulated,
\begin{equation}\begin{split}\label{Convex problem}
&\mathop {\max }\limits_{{\boldsymbol{y}}_k} \text{ }
            \zeta({\boldsymbol{y}}_k){\rm} \\
\text{s.t.    }
            & {S^i_{k,1}}:y_{k}^i \ge 0, \ i = 0,1,2,...,N\\
            & {S_{k,2}}:{y_k^0}{P_{t}} {\rm{ - }}\sum\limits_{{i} = 1}^{N} {y_{k}^i}  \ge 0\\
            & {S_{k,3}}:{R_k^{'}}{\rm{ - }}{R_{t}} \ge 0\\
            & {S_{k,4}}:P_k^c{y_k^0}{\rm{ + }}\sum\limits_{i = 1}^{N} {y_k^i}  - 1 = 0.
\end{split}\end{equation}

 We are now ready to present the following Theorem{, which verifies the equivalence of the optimization problems theoretically.}
 \newtheorem {theorem}{Theorem}
  \begin{theorem}
  The optimization problem in (\ref{Convex problem}) is a convex optimization problem and equivalent to the optimization problem in (\ref{Nonconvex problem}).
  \end{theorem}
  \begin{proof}
  From (\ref{Rate}), (\ref{EE}), and (\ref{Nonconvex problem}), let
\begin{align*}
{f_{k}}\left( {{\boldsymbol{p}}_k^{}} \right)& = \frac{1}{W} {R_{k}}\left( {{\boldsymbol{p}}_k^{}} \right)  \\
&= \sum\limits_{i = 1}^{N} {{{\log }_2}\left( {1 + \frac{{H_{k,k}^iP_k^i}}{{{N_0}{\rm{ + }}H_{0,k}^iP_0^i{\rm{ + }}{\sum\limits_{l \ne k,l = 1}^{K}} {H_{l,k}^iP_l^i} }}} \right)}, \\
{h_k}\left( {{\boldsymbol{p}}_k^{}} \right)&={{P_k^c{\rm{ + }} \frac{1}{\sigma} \sum\limits_{i = 1}^{N} {P_k^i} }}.
\end{align*}
{Taking} the second order derivative of ${f_{k}}\left( {{\boldsymbol{p}}_k^{}} \right)$ and ${h_k}\left( {{\boldsymbol{p}}_k^{}} \right)$ with respect to \(P_{k}^i \), we have
\begin{align*}
\frac{{{\partial ^2}{f_{k}}({{\boldsymbol{p}}_k})}}{{\partial {{(P_{{k}}^i)}^2}}} &= -\frac{{ {{(H_{kk}^i)}^2}}}{{{{\left( {{N_0} + H_{kk}^iP_k^i + H_{0k}^iP_0^i{\rm{ + }}{{{\sum\limits_{l \ne k,l = 1}^{K}} {H_{l,k}^iP_l^i} }} } \right)}^2}}} \\
&< 0, \\
\frac{{{\partial ^2}{h_k}({{\boldsymbol{p}}_k})}}{{\partial {{(P_{{k}}^i)}^2}}} &= 0.
\end{align*}
{We can also readily establish that $\frac{{{\partial ^2}{f_{k}}({{\boldsymbol{p}}_k})}}{{\partial {{P_{{k}}^i} \partial {P_{{k}}^j}}}} = \frac{{{\partial ^2}{h_{k}}({{\boldsymbol{p}}_k})}}{{\partial {{P_{{k}}^i} \partial {P_{{k}}^j}}}} = 0, \ \forall i \ne j.$
Therefore, the Hessian matrix of $f_{k}({{\boldsymbol{p}}_k})$ is negative definite and  \({f_{k}}\left( {{\boldsymbol{p}}_k^{}} \right) \) is concave,
and the Hessian matrix of $h_{k}({{\boldsymbol{p}}_k})$ is positive definite and \({h_k}\left( {{\boldsymbol{p}}_k^{}} \right) \)  is convex.}

From (\ref{Nonconvex problem}), let
\begin{align*}
{l_{k,1}^i}({\boldsymbol{p}}_k)&=-P_{k}^i,\  i = 1,2,...,{N}, \\
{l_{k,2}}({\boldsymbol{p}}_k)&=\sum\limits_{{{i}} = 1}^{N} {P_{k}^i}, \\
{l_{k,3}}({\boldsymbol{p}}_k)& = -\frac{1}{W} {R_{k}}\left( {{\boldsymbol{p}}_k^{}} \right) = -{f_{k}}\left( {{\boldsymbol{p}}_k^{}} \right).
\end{align*}
 {Taking} the second order derivative of \({l_{k,1}^i}({\boldsymbol{p}}_k) \), \({l_{k,2}}({\boldsymbol{p}}_k)\) and \({l_{k,3}}({\boldsymbol{p}}_k)\) with respect to \(P_{k}^i \), we have
\begin{align*}
\frac{{{\partial ^2}{l_{k,1}^i}({\boldsymbol{p}}_k)}}{{\partial {{(P_{k}^i)}^2}}} &= \frac{{{\partial ^2}{l_{k,2}}({\boldsymbol{p}}_k)}}{{\partial {{(P_{k}^i)}^2}}} = 0, \\
\frac{{{\partial ^2}{l_{k,3}}({{\boldsymbol{p}}_k})}}{{\partial {{(P_{k}^i)}^2}}} &= \frac{{ {{\left(H_{kk}^i \right)}^2}}}{{{{\left( {{N_0} + H_{kk}^iP_k^i + H_{0k}^iP_0^i{\rm{ + }}{{{\sum\limits_{l \ne k,l = 1}^{K}} {H_{l,k}^iP_l^i} }} } \right)}^2}}} \\
& > 0.
\end{align*}
{We can also readily establish that $\frac{{{\partial ^2}{l_{k,1}^i}({{\boldsymbol{p}}_k})}}{{\partial {{P_{{k}}^i} \partial {P_{{k}}^j}}}} = \frac{{{\partial ^2}{l_{k,2}}({{\boldsymbol{p}}_k})}}{{\partial {{P_{{k}}^i} \partial {P_{{k}}^j}}}} = \frac{{{\partial ^2}{l_{k,3}}({{\boldsymbol{p}}_k})}}{{\partial {{P_{{k}}^i} \partial {P_{{k}}^j}}}} = 0, \ \forall i \ne j.$
Therefore, the Hessian matrices of \({l_{k,1}^i}({\boldsymbol{p}}_k)\), \({l_{k,2}}({\boldsymbol{p}}_k)\) and \({l_{k,3}}({\boldsymbol{p}}_k)\) are all positive definite,
and \({l_{k,1}^i}({\boldsymbol{p}}_k)\), \({l_{k,2}}({\boldsymbol{p}}_k)\) and \({l_{k,3}}({\boldsymbol{p}}_k)\) are all convex.}
Correspondingly, the  sublevel sets \(C_{k,1}\),\(C_{k,2}\), \(C_{k,3}\) of \({l_{k,1}^i}({\boldsymbol{p}}_k)\), \({l_{k,2}}({\boldsymbol{p}}_k)\) and \({l_{k,3}}({\boldsymbol{p}}_k)\) are convex \cite{Boyd},
and the set \(C_k= {C_{k,1}} \cap {C_{k,2}} \cap {C_{k,3}}\) is convex.

Consequently, according to \cite{Pfree}, the optimization problem in (\ref{Convex problem}) is a convex optimization problem and equivalent to the optimization problem in (\ref{Nonconvex problem}).
  \end{proof}

\subsubsection{Elimination of Constraints}

{In general, the exterior penalty function method \cite{Jiang, Wang} can be applied to solve convex optimization problems with equality constraints and inequality constraints.
However, the solution obtained  may not satisfy the constraints.
On the other hand, the interior penalty function method \cite{Wright} can only solve convex optimization problems with inequality constraints,
but the solution obtained can {always} satisfy the constraints.}
We can see from (\ref{Convex problem}) that the equivalent optimization problem includes both the inequality constraints \({S^i_{k,1}}\ (i=0,1,2,...,N)\), \(S_{k,2}\), \(S_{k,3}\) and  the equality constraint \(S_{k,4}\).
Therefore, we propose to use the mixed penalty function method to transform the constrained convex optimization problem in (\ref{Convex problem}) into an unconstrained one.
The inequality constraints  \({S^i_{k,1}} (i=1,2,...,N)\), \(S_{k,2}\) and \(S_{k,3}\) are eliminated by introducing a logarithmic barrier function based on the interior penalty function  method,
and the equality constraint \(S_{k,4}\) is eliminated by introducing a quadratic penalty function based on the exterior penalty function method.
{Correspondingly, both the disadvantage of the exterior penalty function method and that of the interior penalty function method can be avoided.}

Let
\begin{align*}
\varphi_1({\boldsymbol{y}}_k) &= \sum\limits_{i = 0}^{N} {\ln  y_k^i},\\
\varphi_2({\boldsymbol{y}}_k) &= \ln \left({y_k^0}{P_{t}}  - \sum\limits_{{\rm{i}} = 1}^N {y_{\rm{k}}^i} \right), \\
\varphi_3({\boldsymbol{y}}_k) &= \ln \left({R_{k}^{'}} - {R_t}\right).
\end{align*}
Define the logarithmic barrier function with respect to the inequality constraints as follows,
\begin{equation}\begin{split}\label{3-01}
{\phi _{\text{ie}}}({\boldsymbol{y}}_k)
        &=  - \varphi_1({\boldsymbol{y}}_k)   - \varphi_2({\boldsymbol{y}}_k) - \varphi_3({\boldsymbol{y}}_k).
\end{split}\end{equation}
According to \cite{Boyd} and \cite{ Wright}, it can  be readily  established that $\varphi_1({\boldsymbol{y}}_k)$, $\varphi_2({\boldsymbol{y}}_k)$ and $\varphi_3({\boldsymbol{y}}_k)$ are concave.
Correspondingly, \({\phi _{\text{ie}}}({\boldsymbol{y}}_k)\) is convex.
Define the quadratic penalty function with respect to the quality constraint as follows,
\begin{equation}\begin{split}\label{3-01}
{\phi _{\text{e}}}({\boldsymbol{y}}_k) = {\left(P_k^c{y_k^0}{\rm{ + }}\sum\limits_{i = 1}^N {y_k^i}  - 1 \right)^2}.
\end{split}\end{equation}
Obviously, \({\phi _{\text{e}}}({\boldsymbol{y}}_k)\) is convex.

Let ${\mu}_{\text{ie}}$ and ${\mu}_{\text{e}}$ denote the positive penalty factors.
{${\mu}_{\text{ie}}$ should be set as small as possible and ${\mu}_{\text{e}}$ should be set as large as possible. }
Define
\begin{equation}\begin{split}\label{Unconstrained function}
{\psi }({\boldsymbol{y}}_k) =  - \zeta ({\boldsymbol{y}}_k)+ {{\mu}_{\text{ie}}} {\phi _{\text{ie}}}({\boldsymbol{y}}_k){\rm{ + }}{{{\mu}_{\text{e}}} }{\phi _{\text{e}}}({\boldsymbol{y}}_k).
\end{split}\end{equation}
Then, the constrained convex optimization problem in (\ref{Convex problem}) can be transformed into an equivalent noncontrained convex optimization problem by using the mixed penalty method as follows,
\begin{equation}\begin{split}\label{Unconstrained convex problem}
\mathop {\min }\limits_{{\boldsymbol{y}}_k} {\psi }({\boldsymbol{y}}_k).
\end{split}\end{equation}

{
Correspondingly, we have the following Theorem.
  \begin{theorem}
  The non-constrained convex optimization problem in (\ref{Unconstrained convex problem}) is equivalent to the constrained convex optimization problem in (\ref{Convex problem}).
  \end{theorem}
  \begin{proof}
  From \emph{Theorem 1}, we know that the optimization problem in (\ref{Convex problem})  is a convex optimization.
  According to the barrier method in \cite{Boyd} and  \emph{Theorem 4} in \cite{Wright}, the optimization problem in (\ref{Convex problem}) can be transformed into an equivalent convex optimization problem as follows,
  \begin{equation}\begin{split}\label{Equation Convex problem}
\mathop {\max }\limits_{{\boldsymbol{y}}_k} &\text{ }
            {\psi '}({\boldsymbol{y}}_k) =  - \zeta ({\boldsymbol{y}}_k)+ {{\mu}_{\text{ie}}} {\phi _{\text{ie}}}({\boldsymbol{y}}_k){\rm{ }} \\
\text{s.t.    }
            & P_k^c{y_k^0}{\rm{ + }}\sum\limits_{i = 1}^{N} {y_k^i}  - 1 = 0.
\end{split}\end{equation}
Then, by introducing the augmented Lagrangian terms in \cite{bloom2014exterior}, the convex optimization problem in (\ref{Equation Convex problem}) can be transformed into an equivalent unconstrained convex optimization problem in (\ref{Unconstrained convex problem}).
Therefore, the non-constrained convex optimization problem in (\ref{Unconstrained convex problem})  is equivalent to the constrained convex optimization problem in (\ref{Convex problem}).
  \end{proof}
}

Define
\begin{equation*}
{\boldsymbol{y}}_k^*(n) = \left[ y_k^{0,*}(n),y_k^{1,*}(n), \cdots, y_k^{i,*}(n), \cdots, y_k^{N,*}(n) \right]^T,
\end{equation*}
and let ${\boldsymbol{y}}_k^*(n)$ denote the solution of the above unconstrained optimization problem
for the $n$th iteration in the non-cooperative power control game,
which can readily be obtained by using the gradient method.
{Then, the optimal transmit power can be calculated as follows,}
\begin{equation}\begin{split}\label{}
P_k^{i}(n) = \frac{{y_k^{i,*}}(n)} {{y_k^{0,*}}(n)}, \ i=1,2,...,N.
\end{split}\end{equation}

{It should be pointed out here that the constraints have been removed by introducing only two parameters in our proposed approach.
Correspondingly, the computational complexity of the problem solving can be reduced greatly.}

\subsection{The {Existence} and Uniqueness of the NE}

NE offers a predictable and stable outcome about the transmit power strategy that each SBS will choose.
For our considered non-cooperative power control game, we have the following theorem.
  \begin{theorem}
  There exists one and only one NE in the non-cooperative power control game \(G = [\mathcal{K} ,\{ {{\boldsymbol{p}}_k}\} ,\{ {u_k}( \cdot )\} ]\).
  \end{theorem}
  \begin{proof}
  It has been established that the optimization problems in (\ref{Nonconvex problem}), (\ref{Convex problem}) and  (\ref{Unconstrained convex problem}) are equivalent.
 Accordingly, we can treat \({\psi }({\boldsymbol{y}}_k)\)  as the payoff function of the $k$th SBS in the non-cooperative power control game \(G = [\mathcal{K} ,\{ {{\boldsymbol{p}}_k}\} ,\{ {u_k}( \cdot )\} ]\).
 From the above description, the payoff function \({\psi }({\boldsymbol{y}}_k)\) is continuous with respect to $y_k^i$, and  \({\psi }({\boldsymbol{y}}_k)\) is convex.
 Therefore, NE exists in the considered non-cooperative power control game.

{
According to the definition of $\varphi_1({\boldsymbol{y}}_k)$,}
we can readily {establish} that
\begin{align*}
\frac{{{\partial ^2}{\varphi_1({\boldsymbol{y}}_k)}}}{{{{\left( {\partial y_{k}^i} \right)}^2}}} &=  - \frac{1}{{{{\left( {y_{k}^i} \right)}^2}}}<0, \\
\frac{{{\partial ^2}{\varphi_1({\boldsymbol{y}}_k)}}}{{\partial y_{k}^i\partial y_{k}^j}} &= 0, \ \forall i \ne j.
\end{align*}
{
Therefore, the Hessian matrix of \(\varphi_1({\boldsymbol{y}}_k)\) is negative definite
and \(\varphi_1({\boldsymbol{y}}_k)\) is strictly concave.
It is already known that $\varphi_2({\boldsymbol{y}}_k)$ and $\varphi_3({\boldsymbol{y}}_k)$ are concave, and \({\phi _{\text{e}}}({\boldsymbol{y}}_k)\) are convex.}
Correspondingly, \({\phi _{\text{ie}}}({\boldsymbol{y}}_k)\) is strictly convex,
\({\psi}({\boldsymbol{y}}_k)\) is strictly convex,
and the optimization problem in (\ref{Unconstrained convex problem}) has a unique optimal solution.
Therefore, a unique NE exists  in the proposed non-cooperative power control game.
\end{proof}

{
It should be pointed out here that NE is a combination of strategies such that each participant's strategy at the same time is the optimal response to the other participants' strategies,
and that NE, even when it is unique,  does not mean that all the participants of the game have achieved their global optimum.
One of the metrics used to measure how the efficiency of a system degrades due to the selfish behavior of its participants in economics and game theory is the so-called price of anarchy \cite{Nisan, Koutsoupias},
which is defined as the ratio between the worst NE point and the social optima.
According to \cite{Roughgarden}, research over the past seventeen years has provided an encouraging counterpoint to this widespread equilibrium inefficiency: in a number of interesting application domains, game-theoretic equilibria provably approximate the optimal outcome.
That is, the price of anarchy is close to 1 in many interesting games.
}

\section{Simulation Results}

In this section, the performance of the proposed non-cooperative power control scheme is evaluated via simulations.
In the simulation, we consider the same number of SUEs in each small cell with a path loss function
\( H= \kappa d^{-\chi}\) \cite{Jiang}, where \(d\) denotes the distance between the BS and the UE, $\kappa=10^{-1}$ and $\chi=4$  the path loss constant and  path loss exponent, respectively.
In the simulations,  the radiuses of the macrocell and small cells are set to be 1000 meters and 100 meters, respectively.
The noise spectral density is set to be -174dBm/Hz.

\begin{figure}[!b]
\centering 
\includegraphics[width=0.45\textwidth]{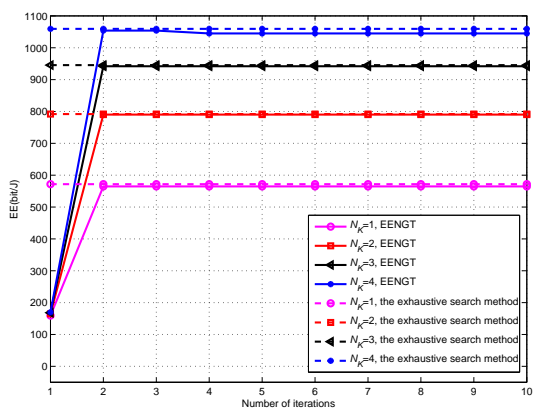}
\captionstyle{mystyle2}
\caption{EE versus the  number of iterations with different $N_K$.}
\label{fig2}
\end{figure}
\begin{figure}[!b]
\centering 
\includegraphics[width=0.45\textwidth]{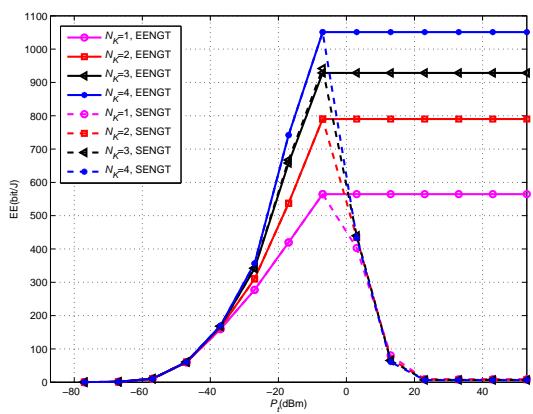}
\captionstyle{mystyle2}
\caption{EE versus the maximum transmit power \(P_t\) with different $N_K$.}
\label{fig3}
\end{figure}

In Fig. \ref{fig2}, we illustrate the EE of the proposed scheme {optimizing the EE based on non-cooperative game theory (hereafter referred to as EENGT) } versus the number of iterations with different \(N_K\) for \(K=2\), {$R_t=3$ bit/s} and \(P_t=20\) dBm.
Also illustrated in the figure as a performance benchmark is the EE of the exhaustive search method,
using which the EE of the system in (\ref{EEs}) is  maximized.
We can observe that the larger the number of SUEs, the larger the EE of the proposed scheme.
The reason is that different SUEs in the considered same cell do not crosstalk each other,
and the SUEs can be allocated the suitable power levels.
Then, multiuser diversity can be exploited.
We can also observe that only several iterations are required for the proposed scheme to converge to the NE.
It can be observed that the EE of the proposed scheme approaches the EE of the exhaustive search method for the considered simulation scenario.
However, the complexity can be decreased dramatically by using the proposed scheme.
Note here that the complexity of the proposed scheme is $\mathcal{O}(KN^2)$, whereas the exhaustive search method has an exponential complexity with respect to the number of RBs and the number of cells.


%

In Fig. \ref{fig3}, we show the EE versus the maximum transmit power \(P_t\) with different \(N_K\)  for \(K=2\) { and $R_t=3$ bit/s}.
Also included in the figure is the EE of the power control scheme optimizing SE based on non-cooperative game theory (hereafter referred to as SENGT) in \cite{Huang}.
For the {SENGT}, SE instead of EE is maximized.
It can be observed that the EE of the  EENGT increases with $P_t$  when $P_t$ is smaller than a certain threshold value,
and that the EE almost remains  constant  when $P_t$ is large enough.
We can readily observe that the  EENGT, which optimizes EE,
provides an obvious performance improvement in terms of EE over the {SENGT} in \cite{Huang}, which optimizes SE.
The reason is that the latter scheme uses excess power to increase the SE by sacrificing the EE, especially in the high transmit power region.

In Fig. \ref{fig4}, we show the SE versus the maximum transmit power \(P_t\) with different \(N_K\) for \(K=2\) and { $R_t=3$ bit/s}.
We compare the system performance of the  EENGT again with the {SENGT} in \cite{Huang}.
It can be  obviously observed that the SE of the  EENGT increases with the maximum transmit power in the low transmit power region,
and that the SE remains constant in the high transmit power region.
The reason is that the  EENGT clips the transmit power at the SBSs to maximize the system EE.
It can also be observed that the {SENGT} achieves a higher SE than the  EENGT.
The reason is that the former scheme consumes all the available transmit power in all the considered transmit power region.
However, the SE of the baseline scheme comes at the expense of lower EE.

{We can observe from Fig. \ref{fig3} and Fig. \ref{fig4} that both the EE and SE of our proposed scheme increase with the maximum transmit power  in the low transmit power region,
and  they almost remain  constant in the high transmit power region.
This gives us the insight that only appropriate, rather than   exorbitant, transmit power needs to be allocated  for each SBS to achieve its maximum utility and reach the NE point.
}

\begin{figure}[!b]
\centering 
\includegraphics[width=0.45\textwidth]{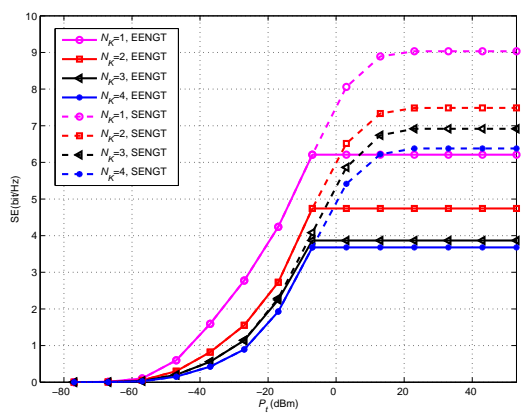}
\captionstyle{mystyle2}
\caption{SE versus the maximum transmit power \(P_t\) with different $N_K$.}
\label{fig4}
\end{figure}

\section{Conclusions}

In this paper, we have formulated a non-cooperative power control game for small cells underlaying a macro cellular network.
By exploiting the properties of parameter-free fractional programming, the concept of perspective function, and the mixed penalty function,
an energy efficient power control scheme has been proposed to maximize the EE.
Simulation results have shown that significant improvements in terms of EE are achieved by using the proposed scheme.

%


\begin{biography}[{\includegraphics[width=1in,height
=1.25in,clip,keepaspectratio]{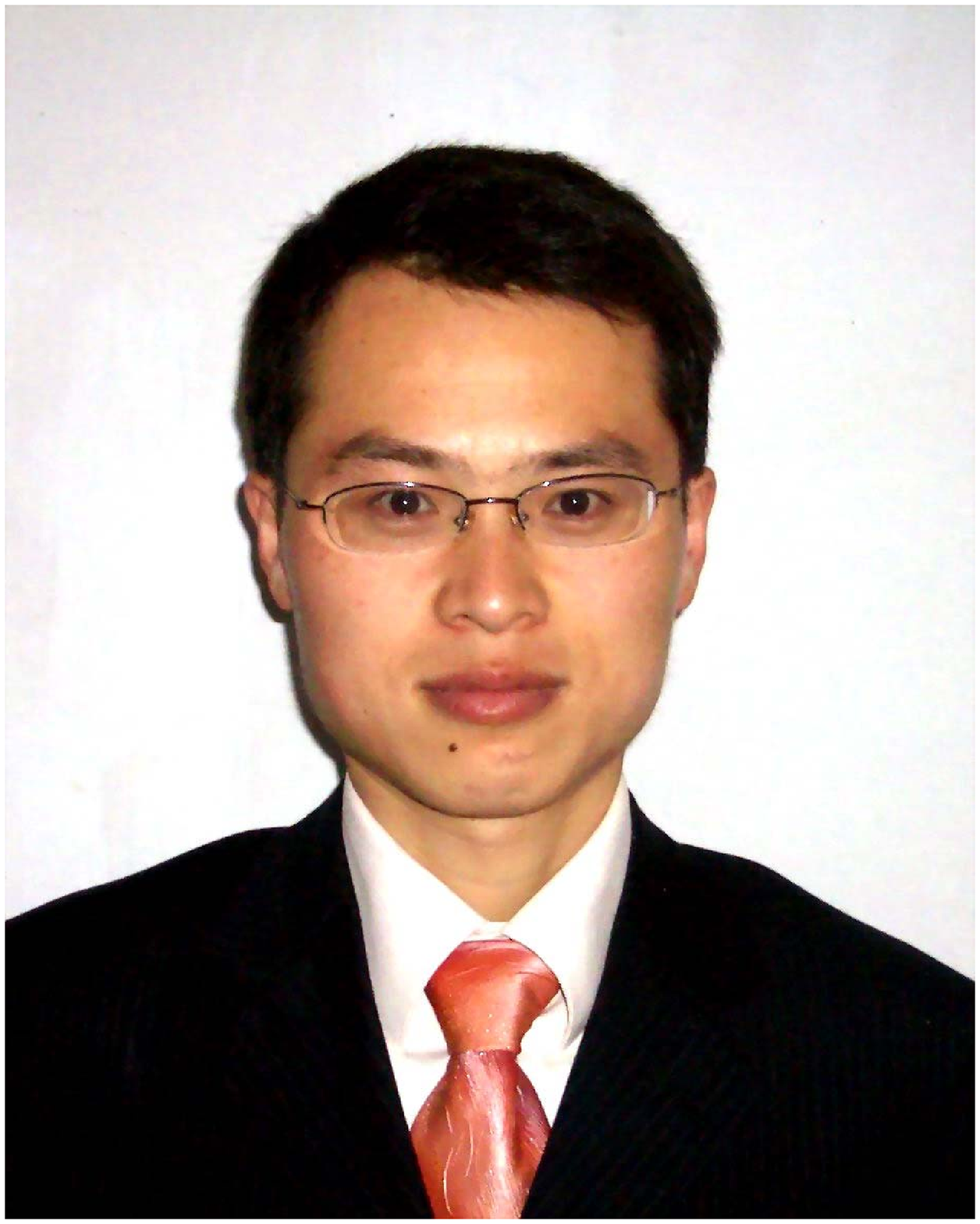}}]
{Yanxiang Jiang (S'03-M'07)}
received the B.S. degree
in electrical engineering from Nanjing University,
Nanjing, China, in 1999 and the M.S. and
Ph.D. degrees in communications and information
systems from Southeast University, Nanjing, China,
in 2003 and 2007, respectively.

From December 2013 to December 2014, he was
a Visiting Scholar with the Signals and Information
Group, Department of Electrical and Computer Engineering,
University of Maryland at College Park,
College Park, MD, USA, with Prof. K. J. Ray Liu. He
is currently with the faculty of the National Mobile Communications Research
Laboratory, Southeast University, Nanjing, China. His current research interests include
signal processing, and wireless communications.

Dr. Jiang received the Distinguished Graduated Student Award from
Nanjing University, Nanjing, China, and the Outstanding Distinguished Doctoral
Dissertation from Southeast University.
\end{biography}


\begin{biography}[{\includegraphics[width=1in,height
=1.25in,clip,keepaspectratio]{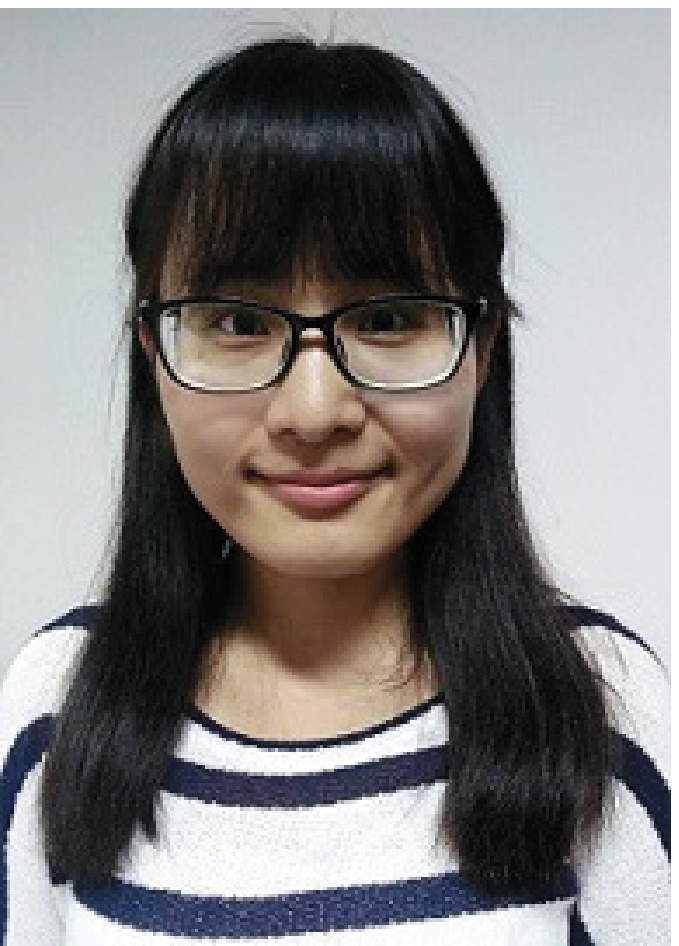}}]
{Ningning Lu}
is currently pursuing the M.S. degree in communications and information systems from Southeast University, Nanjing, China.
Her research interests include radio resource management, and mobile communication systems.
\end{biography}


\begin{biography}[{\includegraphics[width=1in,height
=1.25in,clip,keepaspectratio]{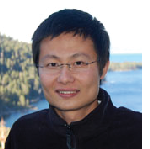}}]
{Yan Chen (SM'14)}
received the Bachelor's degree
from University of Science and Technology of China
in 2004, the M.Phil. degree from Hong Kong University
of Science and Technology (HKUST) in 2007,
and the Ph.D. degree from University of Maryland
College Park in 2011. His current research interests
are in data science, network science, game theory, social
learning and networking, as well as signal processing
and wireless communications.

Dr. Chen is the recipient of multiple honors
and awards including best paper award from IEEE
GLOBECOM in 2013, Future Faculty Fellowship and Distinguished Dissertation
Fellowship Honorable Mention from Department of Electrical and
Computer Engineering in 2010 and 2011, respectively, Finalist of Deans
Doctoral Research Award from A. James Clark School of Engineering at
the University of Maryland in 2011, and Chinese Government Award for
outstanding students abroad in 2011.

\end{biography}


\begin{biography}[{\includegraphics[width=1in,height
=1.25in,clip,keepaspectratio]{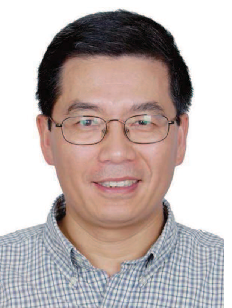}}]
{Fuchun Zheng (M'95-SM'99)}
obtained the BEng (1985) and MEng (1988) degrees in radio engineering from Harbin Institute of Technology, China,
and the PhD degree in Electrical Engineering from the University of Edinburgh, UK, in 1992.

From 1992 to 1995, he was a post-doctoral research associate with the University of Bradford, UK, Between May 1995 and August 2007, he was with Victoria University, Melbourne, Australia, first as a lecturer and then as an associate professor in mobile communications.  He was with the University of Reading, UK, from September 2007 to July 2016 as a Professor (Chair) of Signal Processing. He has also been a ``1000 talents'' adjunct professor with Southeast University, China, since 2010. Since August 2016, he has been with Harbin Institute of Technology (Shenzhen), China and the University of York, UK as a specially appointed professor. He has been awarded two UK EPSRC Visiting Fellowships - both hosted by the University of York (UK): first from August 2002 to July 2003 and then from August 2006 to July 2007. Over the past 20 years, Dr Zheng has also carried out many government and industry sponsored research projects - in Australia, the UK, and China. He has been both a short term visiting fellow and a long term visiting research fellow with British Telecom, UK. Dr Zheng's current research interests include signal processing for communications, multiple antenna systems, and green communications.

He has been an active IEEE member since 1995. He was an editor (2001-2004) of IEEE TRANSACTIONS ON WIRELESS COMMUNICATIONS. In 2006, Dr Zheng served as the general chair of IEEE VTC 2006-S in Melbourne, Australia (www.ieeevtc.org/vtc2006spring) - the first ever VTC held in the southern hemisphere in VTC’s history of six decades. He was the executive TPC Chair for VTC 2016-S in Nanjing, China (the first ever VTC held in mainland China: www.ieeevtc.org/vtc2016spring).

\end{biography}


\begin{biography}[{\includegraphics[width=1in,height
=1.25in,clip,keepaspectratio]{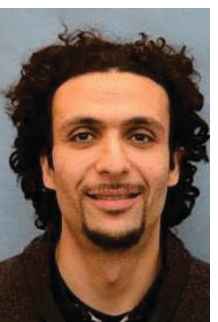}}]
{Mehdi Bennis (S'07-AM'08-SM'15)}
received the M.Sc. degree in electrical engineering jointly from
the EPFL, Lausanne, Switzerland and the Eurecom
Institute, Biot, France, in 2002, and the Ph.D. degree
in spectrum sharing for future mobile cellular systems,
in 2009. From 2002 to 2004, he worked as
a Research Engineer with IMRA-EUROPE investigating
adaptive equalization algorithms for mobile
digital TV. In 2004, he joined the Centre for Wireless
Communications (CWC), University of Oulu, Oulu,
Finland, as a Research Scientist. In 2008, he was a
Visiting Researcher at the Alcatel-Lucent Chair on flexible radio, SUPELEC.
He has coauthored one book and authored more than 100 research papers
in international conferences, journals, and book chapters. His research interests
include radio resource management, heterogeneous networks, game theory
and machine learning in 5G networks, and beyond. He serves as an Editor
for the IEEE TRANSACTIONS ON WIRELESS COMMUNICATIONS. He was
the recipient of the prestigious 2015 Fred W. Ellersick Prize from the IEEE
Communications Society.
\end{biography}


\begin{biography}[{\includegraphics[width=1in,height
=1.25in,clip,keepaspectratio]{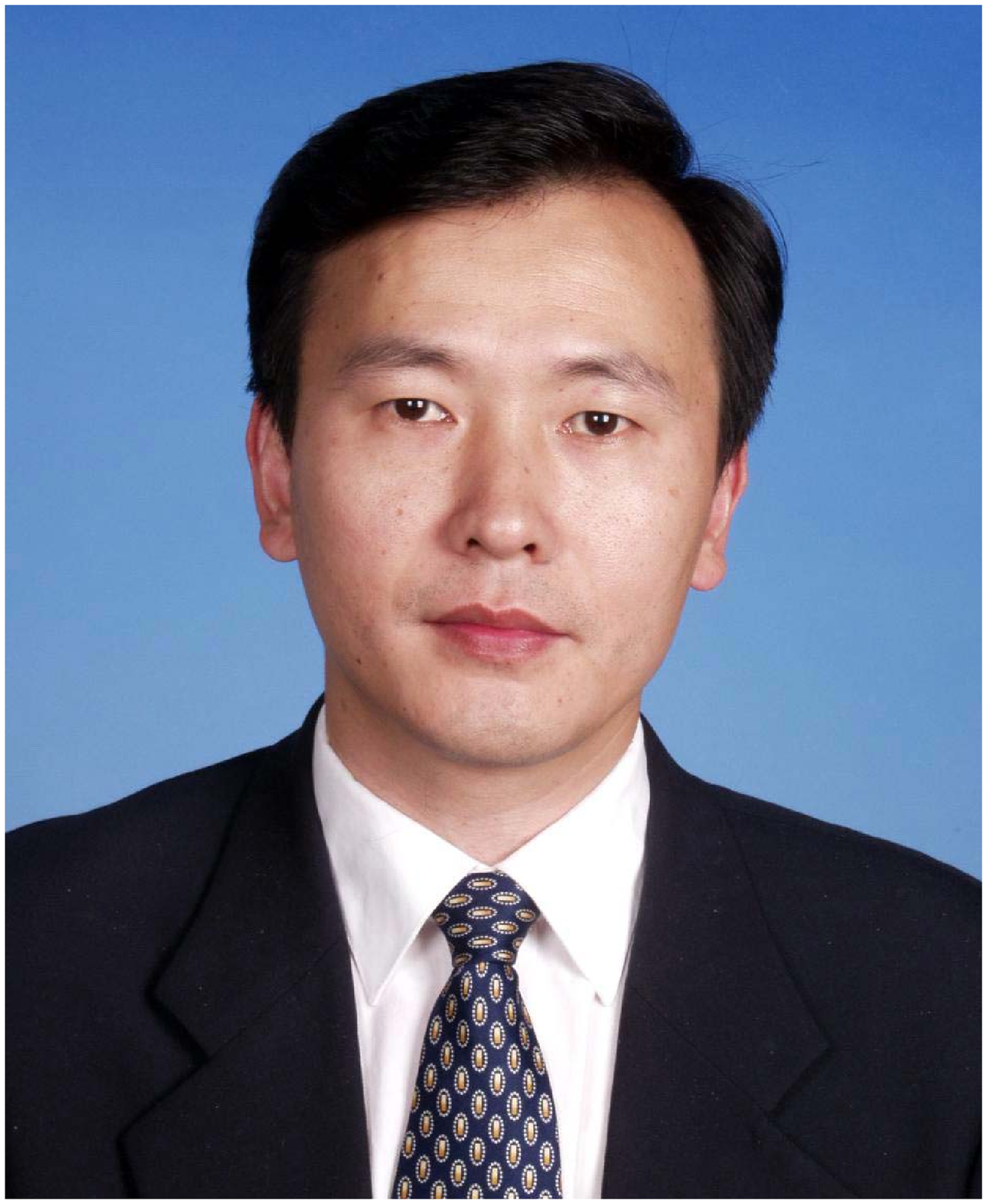}}]
{Xiqi Gao (SM'07-F'15)} received the Ph.D. degree
in electrical engineering from Southeast University,
Nanjing, China, in 1997. In April 1992, he joined the
Department of Radio Engineering, Southeast University,
where he has been a Professor of information
systems and communications since May 2001. From
September 1999 to August 2000, he was a Visiting
Scholar at Massachusetts Institute of Technology,
Cambridge, MA, USA, and Boston University,
Boston, MA. From August 2007 to July 2008, he visited
Darmstadt University of Technology, Darmstadt,
Germany, as a Humboldt Scholar. His current research interests include broadband
multicarrier communications, MIMO wireless communications, channel
estimation and turbo equalization, and multirate signal processing for wireless
communications.

Dr. Gao served as an Editor of the IEEE TRANSACTIONS ON WIRELESS
COMMUNICATIONS from 2007 to 2012. From 2009 to 2013, he served as an
Editor of the IEEE TRANSACTIONS ON SIGNAL PROCESSING. He is currently
serving as an Editor of the IEEE TRANSACTIONS ON COMMUNICATIONS. He
received the Science and Technology Awards of the State Education Ministry
of China in 1998, 2006, and 2009; the National Technological Invention
Award of China in 2011; and the 2011 IEEE Communications Society Stephen
O. Rice Prize Paper Award in the field of communications theory.

\end{biography}


\begin{biography}[{\includegraphics[width=1in,height
=1.25in,clip,keepaspectratio]{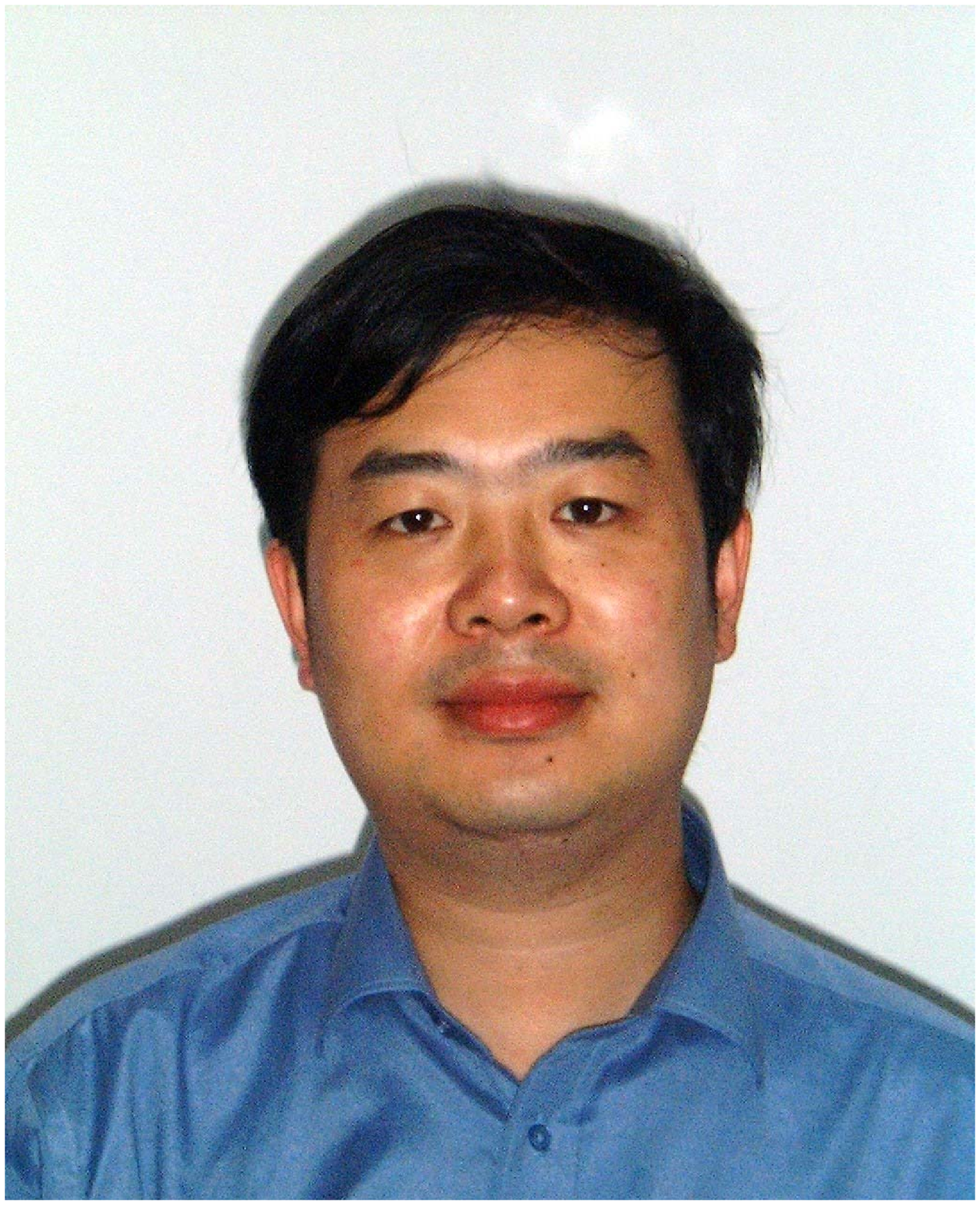}}]
{Xiaohu You (SM'11-F'12)} received the B.S., M.S., and Ph.D. degrees in electrical engineering
from Southeast University, Nanjing, China, in 1982, 1985, and 1988,
respectively. Since 1990, he has been working with National Mobile Communications
Research Laboratory at Southeast University, where he holds the
ranks of professor and director. He is the
Chief of the Technical Group of China 3G/B3G
Mobile Communication R$\&$D Project. His research
interests include mobile communications, adaptive
signal processing, and artificial neural networks, with applications to communications
and biomedical engineering.

Dr. You was a recipient of the Excellent Paper Prize from the China Institute
of Communications in 1987; the Elite Outstanding Young Teacher award from
Southeast University in 1990, 1991, and 1993; and the National Technological
Invention Award of China in 2011. He was also a recipient of the 1989 Young
Teacher Award of Fok Ying Tung Education Foundation, State Education
Commission of China.
\end{biography}

\end{document}